\documentclass[proceedings]{stacs}
\stacsheading{2010}{59-70}{Nancy, France}
\firstpageno{59}

\usepackage{amsmath}
\usepackage{amssymb}
\usepackage{amsthm}

\newcommand{\avg}[2]{\mathop{\textbf{E}}_{#1}[#2]}
\newcommand{\poly}[1]{\mathop{\mathrm{poly}}(#1)}

\newcommand{\bra}[1]{\{#1\}}

\newcommand{\setcond}[2]{\left\{#1\: \middle|\: #2\right\}}

\newcommand{\NC}{\mbox{\rm NC}}

\newcommand{\AC}{\mbox{\rm AC}}

\newcommand{\field}{\mathbb{F}}

\newcommand{\complex}{\mathbb{C}}
\newcommand{\naturals}{\mathbb{N}}
\newcommand{\integer}{\mathbb{Z}}

\newcommand{\intersect}{\cap}
\newcommand{\mcG}{\mathcal{G}}
\newcommand{\mcH}{\mathcal{H}}
\newcommand{\mcK}{\mathcal{K}}

\newcommand{\mc}[1]{\mathcal{#1}}
\newcommand{\Sym}{\mathrm{Sym}}

\newcommand{\F}{\mathbb{F}}

\newcommand{\squishlist}{
 \begin{list}{$\bullet$}
  { \setlength{\itemsep}{0pt}
     \setlength{\parsep}{3pt}
     \setlength{\topsep}{3pt}
     \setlength{\partopsep}{0pt}
     \setlength{\leftmargin}{1.5em}
     \setlength{\labelwidth}{1em}
     \setlength{\labelsep}{0.5em} } }

\newcommand{\squishlisttwo}{
 \begin{list}{$\bullet$}
  { \setlength{\itemsep}{0pt}
     \setlength{\parsep}{0pt}
    \setlength{\topsep}{0pt}
    \setlength{\partopsep}{0pt}
    \setlength{\leftmargin}{2em}
    \setlength{\labelwidth}{1.5em}
    \setlength{\labelsep}{0.5em} } }

\newcommand{\squishend}{
  \end{list}  }

\begin{document}

\title{The Remote Point Problem, Small Bias Spaces, and Expanding
Generator Sets}
\author{V.~Arvind}{V.~Arvind}
\author{S.~Srinivasan}{Srikanth Srinivasan}
\address{%
The Institute of Mathematical Sciences, 
\newline C.I.T. Campus, Chennai  600 113, India.}
\email[V.~Arvind]{arvind@imsc.res.in} 
\email[Srikanth Srinivasan]{srikanth@imsc.res.in} 

\subjclass{Algorithms and Complexity Theory}

\begin{abstract}
  Using $\varepsilon$-bias spaces over $\F_2$, we show that the Remote
  Point Problem (RPP), introduced by Alon et al \cite{APY09}, has an
  $\NC^2$ algorithm (achieving the same parameters as \cite{APY09}).
  We study a generalization of the Remote Point Problem to groups: we
  replace $\F_2^n$ by $\mcG^n$ for an arbitrary fixed group $\mcG$.
  When $\mcG$ is Abelian we give an $\NC^2$ algorithm for RPP, again
  using $\varepsilon$-bias spaces. For nonabelian $\mcG$, we give a
  deterministic polynomial-time algorithm for RPP. We also show the
  connection to construction of expanding generator sets for the group
  $\mcG^n$. All our algorithms for the RPP achieve essentially the same
  parameters as \cite{APY09}.
\end{abstract}

\maketitle

\section{Introduction}

Valiant, in his celebrated work \cite{Val} on circuit lower bounds for
computing linear transformations $A:\F^n\longrightarrow \F^m$ for a
field $\F$, initiated the study of rigid matrices. If explicit rigid
matrices of certain parameters can be constructed it would result in
superlinear lower bounds for logarithmic depth linear circuits over
$\F$. This problem and the construction of such rigid matrices has
remained elusive for over three decades.

Alon, Panigrahy and Yekhanin \cite{APY09} recently proposed a problem
that appears to be of intermediate difficulty. Given a subspace $L$ of
$\F_2^n$ by its basis and a number $r\in [n]$ as input, the problem is
to compute in deterministic polynomial time a point $v\in\F_2^n$ such
that $\Delta(u,v)\geq r$ for all $u\in L$, where $\Delta(u,v)$ is the
Hamming distance. They call this the \emph{Remote Point Problem}.  The
point $v$ is said to be $r$-far from the subspace $L$.

Alon et al \cite{APY09} give a nice polynomial time-bounded (in $n$)
algorithm for computing a $v\in \F_2^n$ that is $c\log n$-far from a
given subspace $L$ of dimension $n/2$ and $c$ is a fixed constant. For
$L$ such that $\dim(L)=k<n/2$ they give a polynomial-time algorithm
for computing a point $v\in \F_2^n$ that is $\frac{cn\log k}{k}$-far
from $L$.

\subsubsection*{\bf Results of this paper}

In \cite{AS09a} we recently investigated the problem of proving circuit
lower bounds in the presence of help functions. Specifically, one of
the problems we consider is proving lower bounds for constant-depth
Boolean circuits which can take a given set of (arbitrary) help
functions $\{h_1,h_2,\cdots,h_m\}$ at the input level, where
$h_i:\{0,1\}^n\longrightarrow \{0,1\}$ for each $i$. Proving explicit
lower bounds for this model would allow us to separate EXP from the
polynomial-time many-one closure of nonuniform $\AC^0$. We show that
it suffices to find a polynomial-time solution to the Remote Point
Problem for parameters $k=2^{(\log\log n)^c}$ and
$r=\frac{n}{2^{(\log\log n)^d}}$ for all constants $c$ and $d$.
Unfortunately, the parameters of the Alon et al algorithm are
inadequate for our application.

However, motivated by this connection, in the present paper we carry
out a more detailed study of the Remote Point Problem as an
algorithmic question. We briefly summarize our results.

\noindent\textbf{1.}~~ The first question we address is whether we can
give a deterministic parallel (i.e.\ NC) algorithm for the problem ---
Alon et al's algorithm is inherently sequential as it is based on the
method of conditional probabilities and pessimistic estimators.

It turns out an element of an $\varepsilon$-bias space for suitably
chosen $\varepsilon$ is a solution to the Remote Point Problem which
gives us an NC algorithm quite easily.

\noindent\textbf{2.}~~ Since the RPP for $\F_2^n$ can be solved 
using small bias spaces, it naturally leads us to address the problem
in a more general group-theoretic setting.

In the generalization we study we will replace $\F_2$ with an
arbitrary fixed finite group $\mcG$ such that $|\mcG|\geq 2$. Hence we
will have the $n$-fold product group $\mcG^n$ instead of the vector
space $\F_2^n$.

Given elements $x = (x_1,x_2,\ldots,x_n), y= (y_1,y_2,\ldots,y_n)$ of
$\mcG^n$, let $\Delta(x,y)=|\{i\mid x_i\neq y_i\}|$. I.e.\
$\Delta(x,y)$ is the \emph{Hamming distance} between $x$ and
$y$. Furthermore, for $S\subseteq \mcG^n$, let $\Delta(x,S)$ denote
$\min_{y\in S}\Delta(x,y)$.

We now define the \emph{Remote Point Problem (RPP) over a finite
  group $\mcG$}. The input is a subgroup $\mcH$ of $\mcG^n$, where $\mcH$
is given by a generating set, and a number $r\in [n]$. The problem is
to compute in deterministic polynomial (in $n$) time an element
$x\in\mcG^n$ such that $\Delta(x,H)> r$. The results we show
in this general setting are the following.

\begin{itemize}
\item[(a)] The Remote Point Problem over any \emph{Abelian group}
  $\mcG$ has an $\NC^2$ algorithm for $r=O(\frac{n\log k}{k})$ and
$k\leq n/2$, where $k=\log_{|\mcG|} |\mcH|$. 

\item[(b)] Over an arbitrary group $\mcG$ the Remote point problem has
  a polynomial-time algorithm for $r=O(\frac{n\log k}{k})$ and $k\leq
  n/2$, where $k=\log_{|\mcG|} |\mcH|$.
\end{itemize}

The parallel algorithm stated in part(a) above is based on
$\varepsilon$-bias space constructions for finite Abelian groups
described in Azar et al \cite{AMN98}.  The sequential algorithm stated
in part(b) above is a group-theoretic generalization of the Alon et al
algorithm for $\F_2^n$ \cite{APY09}.

Due to lack of space, some proofs have been omitted. They may be found
in the full version which has been published as an ECCC report
\cite{AS09b}.

\section{Preliminaries}

Fix a finite group $\mcG$ such that $|\mcG|\geq 2$. Given any
$x\in\mcG^n$, let $wt(x)$ denote the number of coordinates $i$ such
that $x_i \neq 1$, where $1$ is the identity of the group $\mcG$. By
$B(r)$, we will refer to the set of $x\in\mcG^n$ such that $wt(x)\leq
r$. Given a subset $S$ of $\mcG^n$, $B(S,r)$ will denote the
set $S\cdot B(r) = \setcond{sx}{s\in S, x\in B(r)}$. Clearly, for any
$S\subseteq \mcG^n$ and any $x\in\mcG^n$,
$x\in B(S,r)$ if and only if $\Delta(x,S)\leq r$. We say that $x$ is
\emph{$r$-close} to $S$ if $x\in B(S,r)$ and \emph{$r$-far} from $S$
if $x\notin B(S,r)$.

The \emph{Remote Point Problem (RPP) over $\mcG$} is defined to be the
following algorithmic problem:

\begin{itemize}
\item[] INPUT: A subgroup $\mcH$ of $\mcG^n$ (given by its generators)
  and an $r\in\naturals$.
\item[] OUTPUT: An $x\in \mcG^n$ such that $x\notin B(\mcH,r)$.
\end{itemize}

Clearly, there are inputs to the above problem where no solution can
be found. But the input instances of the kind that we will study will
clearly have a solution (in fact, a random point of $\mcG^n$ will be a
solution with high probability). 

Given a subgroup $\mcH$ of $\mcG^n$, denote by $\delta(\mcH)$ the quantity
$\log_{|\mcG|} |\mcH|$. We will call $\delta(\mcH)$ the
\emph{dimension of $\mcH$ in $\mcG^n$}.

We say that the RPP over $\mcG$ has a $(k(n),r(n))$-algorithm  if
there is an efficient algorithm that solves the Remote Point Problem
when given as input a subgroup $\mcH$ of $\mcG^n$ of dimension at most
$k(n)$ and an $r$ that is bounded by $r(n)$. (Here, `efficient' can
correspond to polynomial time or some smaller complexity class.)

A simple counting argument shows that there is a valid solution to the
RPP over $\mcG$ on inputs $(\mcH,r)$ where $\delta(\mcH) + r\leq
n(1-\frac{H(r/n)}{\log |G|} - \varepsilon)$, for any fixed $\varepsilon
> 0$ (where $H(\cdot)$ denotes the binary entropy function). However,
the best known deterministic solution to the RPP -- from \cite{APY09}
-- is a polynomial time $(k,\frac{cn\log k}{k})$-algorithm which works
over $\field_2^n$ (i.e, the group $\mcG$ involved is the additive group
of the field $\field_2$). 
%
%
%
\subsection{Some Group-Theoretic Algorithms}
We introduce basic definitions and review some group-theoretic
algorithms. Let $\Sym(\Omega)$ denote the group of all permutations on
a finite set $\Omega$ of size $m$. In this section we use $G,H$
etc. to denote \emph{permutation groups on $\Omega$}, which are simply
subgroups of $\Sym(\Omega)$.

Let $G$ be a subgroup of $\Sym(\Omega)$. For a subset
$\Delta\subseteq\Omega$ denote by $G_{\bra{\Delta}}$ the
\emph{point-wise stabilizer} of $\Delta$. I.e $G_{\bra{\Delta}}$ is the
subgroup consisting of exactly those elements of $G$ that fix each
element of $\Delta$.

\begin{theorem}[Schreier-Sims]{\rm\cite{Lu93}}\label{schreier-sims}
\begin{enumerate}
\item If a subgroup $G$ of $\Sym(\Omega)$ is given by a generating set
  as input along with the subset $\Delta$ there is a polynomial-time
	(sequential) algorithm for computing a generator set for $G_{\bra{\Delta}}$.
\item If a subgroup $G$ of $\Sym(\Omega)$ is given by a generating set
	as input, then there is a polynomial time algorithm for computing
	$|G|$.
\item Given as input a permutation $\sigma\in\Sym(\Omega)$ and a
  generator set for a subgroup $G$ of $\Sym(\Omega)$, we can test in
  deterministic polynomial time if $\sigma$ is an element of $G$.
\end{enumerate}
\end{theorem}

We are also interested in a special case of this problem which we now
define. A subset $\Gamma\subseteq\Omega$ is an \emph{orbit} of $G$ if
$\Gamma = \setcond{\sigma(i)}{\sigma\in G}$ for some $i\in\Omega$. Any
subgroup $G$ of $\Sym(\Omega)$ partitions $\Omega$ into orbits (called
$G$-orbits).

For a constant $b > 0$, a subgroup $G$ of $\Sym(\Omega)$ is defined to
be a \emph{$b$-bounded permutation group} if every $G$-orbit is of
size at most $b$.

In \cite{MC87}, McKenzie and Cook studied the parallel complexity of
\emph{Abelian} permutation group problems. Specifically, they gave an
$\NC^3$ algorithm for testing membership in an Abelian permutation
group given by a generator set and for computing the order of an
Abelian permutation group. When restricted to $b$-bounded Abelian
permutation groups, the algorithms of \cite{MC87} for these problems
are actually $\NC^2$ algorithms. We formally state their result and
derive a consequence.

\begin{theorem}[\cite{MC87}]\label{thm_Luks_algo}
	There is an $\NC^2$ algorithm for membership testing in a
	$b$-bounded Abelian permutation group $G$ given by a generator set.
\end{theorem}

We now consider problems over $\mcG^n$, for a fixed finite group
$\mcG$. We know from basic group theory that every group $\mcG$ is a
permutation group acting on itself. I.e.\ every $\mcG$ can be seen as a
subgroup of $\Sym(\mcG)$, where $\mcG$ acts on itself by left (or right)
multiplication. Therefore, $\mcG^n$ can be easily seen as a
permutation group on the set $\Omega = \mcG\times [n]$ and hence, $\mcG^n$
can be considered a subgroup of $\Sym(\Omega)$. Furthermore, notice
that each subset $\mcG\times \{i\}$ is an orbit of this group
$\mcG^n$. Hence, $\mcG^n$ is a $b$-bounded permutation group contained
in $\Sym(\Omega)$, where $b=|\mcG|$. Finally, if $\mcG$ is an Abelian
group, then so is this subgroup of $\Sym(\Omega)$. We have the
following lemma as an easy consequence of Theorem~\ref{thm_Luks_algo}.
\begin{lemma}\label{corollary_gp_membership}
  Let $\mcG$ be Abelian. There is an $\NC^2$ algorithm that takes as
  input a generator set for some subgroup $\mcH$ of $\mcG^n$ and an
  $x\in \mcG^n$, and accepts iff $x\in\mcH$.
\end{lemma}
Given any $y = (y_1,y_2,\ldots,y_i)\in\mcG^i$ with $1\leq i\leq n$ and
any $S\subseteq\mcG^n$, let $S_y$ denote the set $\setcond{x\in S}{x_j
  = y_j \text{ for } 1\leq j\leq i}$.

\begin{lemma}\label{corollary_affine_xn}
  Let $\mcG$ be any fixed finite group. There is a polynomial time
  algorithm that takes as input a subgroup $\mcH$ of $\mcG^n$, where
  $\mcH$ is given by generators, and a $y\in\mcG^i$ with $1\leq i\leq
  n$, and computes $|\mcH_y|$.
\end{lemma}

\begin{proof}
  Let $\mcK=\{(x_1,x_2,\ldots,x_n)\in\mcH\mid x_1=x_2=\cdots=x_i=1\}$,
  where $1$ denotes the identity element of $\mcG$. Clearly, $\mcK$ is
  a subgroup of $\mcH$. The set $\mcH_y$, if nonempty, is simply a
  coset of $\mcK$ and thus, we have $|\mcH_y| = |\mcK|$. To check if
  $\mcH_y$ is nonempty, we consider the map
  $\pi_i:\mcG^n\rightarrow\mcG^i$ that projects its input onto its
  first $i$ coordinates; note that $\mcH_y$ is nonempty iff the
  subgroup $\pi_i(\mcH)$ contains $y$, which can be checked in
  polynomial time by point ($3$) of Theorem~\ref{schreier-sims} (here,
  we are identifying $\mcG^n$ with a subgroup of $\Sym(\mcG\times
  [n])$ as above). If $y\notin \pi_i(\mcH)$, the algorithm outputs
  $0$. Otherwise, we have $|\mcH_y|=|\mcK|$ and it suffices to compute
  $|\mcK|$. But $\mcK$ is simply the point-wise stabilizer of the set
  $\mcG\times [i]$ in $\mcH$, and hence $|\mcK|$ can be computed in
  polynomial time by points ($1$) and ($2$) of
  Theorem~\ref{schreier-sims}.
\end{proof}

\section{Expanding Cayley Graphs and the Remote Point Problem}
\label{section_main_lemma}

Fix a group $\mcG$ such that $|\mcG|\geq 2$, and consider an instance
of the RPP over $\mcG$. The main idea that we develop in this section
is that if we have a (symmetric) expanding generator set $S$ for the
group $\mcG^n$ with appropriate expansion parameters then for a
subgroup $\mcH$ of $\mcG^n$ such that $\delta(\mcH)\leq k$ some
element of $S$ will be $r$-far from $H$, for suitable $k$ and $r$.

We review some definitions related to expander graphs (e.g.\ see
the survey of Hoory, Linial, and Wigderson \cite{HLW06}). An
undirected multigraph $G = (V,E)$ is an $(n,d,\alpha)$-graph for
$n,d\in\naturals$ and $\alpha > 0$ if $|V| = n$, the degree of each
vertex is $d$, and the second largest value $\lambda(G)$ from among
the absolute values of eigenvalues of $A(G)$ -- the adjacency matrix
of the graph $G$ -- is bounded by $\alpha d$.

A \emph{random walk} of length $t\in\naturals$ on an
$(n,d,\alpha)$-graph $G = (V,E)$ is the output of the following random
process: a vertex $v_0\in V$ of picked uniformly at random, and for
$0\leq i < t$, if $v_i$ has been picked, then $v_{i+1}$ is obtained by
selecting a neighbour $v_{i+1}$ uniformly at random (i.e a random edge
out of $v_i$ is picked, and $v_{i+1}$ is chosen to be the other
endpoint of the edge); the output of the process is
$(v_0,v_1,\ldots,v_t)$. We now state an important result regarding
random walks on expanders (see \cite[Theorem 3.6]{HLW06} for details).

\begin{lemma}\label{lemma_exp_walk}
  Let $G=(V,E)$ be an $(n,d,\alpha)$-graph and $B\subseteq V$ with
  $|B|\leq \beta n$. Then, the probability that a random walk
  $(v_0,v_1,\ldots,v_t)$ is entirely contained inside $B$ (i.e,
  $v_i\in B$ for each $i$) is bounded by $(\beta+\alpha)^t$.
\end{lemma}

Let $\mcH$ be a group and $S$ a \emph{symmetric} multiset of elements
from $\mcH$. I.e.\ there is a bijection of multisets
$\varphi:S\rightarrow S$ such that $\varphi(s) = s^{-1}$ for each
$s\in S$. We define the Cayley graph $C(\mcH,S)$ to be the (multi)graph $G$
with vertex set $\mcH$ and edges of the form $(x,xs)$ for each
$x\in\mcH$ and each $s\in S$; since $S$ is symmetric, we consider
$C(\mcH,S)$ to be an undirected graph by identifying the edges
$(x,xs)$ and $(xs,(xs)\varphi(s))$, for each $x$ and $s$.

We now show a lemma that will help relate generators of expanding
Cayley graphs on $\mcG^n$ and the RPP over $\mcG$.  In what follows,
let $S$ be a symmetric multiset of elements from $\mcG^n$; let $G$
denote the Cayley graph $C(\mcG^n,S)$; and let $N,D$ denote $|\mcG|^n$
and $|S|$ (counted with repetitions) respectively.

\begin{lemma}\label{lemma_exp_set_subgps}
  Assume $S$ as above is such that $G$ is an $(N,D,\alpha)$-graph,
  where $\alpha \leq \frac{1}{n^d}$, for some fixed $d>0$. Then, given
  any subgroup $\mcH$ of $\mcG^n$ such that $\delta(\mcH)\leq 2n/3$,
  we have $\frac{|S\intersect \mcH|}{|S|}\leq \frac{1}{n^{d/2}}$ for
  large enough $n$ (where the elements of $S\intersect \mcH$ are
  counted with repetitions).
\end{lemma}

\begin{proof}
  Let $S' = S\intersect\mcH$ and let $\eta = |S'|/|S|$. We want an
  upper bound on $\eta$.  Consider a random walk
  $(x_0,x_1,\ldots,x_t)$ of length $t$ on the graph $G$ (the exact
  value of $t$ will be fixed later). Let $\mathcal{B}$ denote the
  following event: there is a $y\in\mcG^n$ such that all the vertices
  $x_0,x_1,\ldots,x_t$ are all contained in the coset $y\mcH$ of
  $\mcH$.  Let $p$ denote the probability that $\mathcal{B}$ occurs.

  We will first lower bound $p$. At each step of the random walk, a
  random $s_i\in S$ is chosen and $x_{i+1}$ is set to $x_is_i$. If
  these $s_i$ all happen to belong to $S'$, then the cosets $x_i\mcH$
  and $x_{i+1}\mcH$ are the same for all $i$ and hence, the event
  $\mathcal{B}$ does occur. Hence, $p \geq \eta^{t}$.

  We now upper bound $p$. Fix any coset $y\mcH$ of the subgroup
  $\mcH$. Since the dimension of $\mcH$ in $\mcG^n$ is bounded by
  $2n/3$, we have $|y\mcH| = |\mcH| \leq |\mcG|^{2n/3} \leq
  2^{-n/3}|\mcG^n|$. That is, the coset $y\mcH$ is a very small subset
  of $\mcG^n$. Applying Lemma \ref{lemma_exp_walk}, we see that the
  probability that the random walk $(x_0,x_1,\ldots,x_t)$ is
  completely contained inside this coset is bounded by $(2^{-n/3} +
  n^{-d})^t \leq \frac{2^t}{n^{dt}}$, for large enough $n$. As the
  total number of cosets of $\mcH$ is bounded by $|\mcG|^n$, an
  application of the union bound tells us that $p$ is upper bounded by
  $|\mcG|^n\frac{2^t}{n^{dt}}\leq \frac{|\mcG|^{n+t}}{n^{dt}}$.
  Setting $t = \frac{2n}{d\log_{|G|}n -2}$ we see that $p$ is at most
  $\frac{1}{n^{dt/2}}$.

  Putting the upper and lower bounds together, we see that $\eta^t
  \leq \frac{1}{n^{dt/2}}$ and hence, $\eta \leq \frac{1}{n^{d/2}}$.
  This completes the proof.
\end{proof}

We follow the structure of the algorithm for the RPP over $\field_2$
in \cite{APY09}. We first describe their $(n/2,c\log n)$-algorithm for
the RPP, followed by our own algorithm. We then describe how they
extend this algorithm to a $(k,\frac{cn\log k}{k})$-algorithm for any
$k\leq n/2$; the same procedure works for our algorithm also.

The $(n/2,c\log n)$-algorithm proceeds as follows. On an input
instance consisting of a subgroup $V$ (which is a subspace of
$\field_2^n$) of dimension at most $n/2$ and an $r\leq c\log n$,

\begin{enumerate}
\item The algorithm first computes a collection of $m = n^{O(c)}$
  subspaces $V_1,V_2,\ldots,V_m$, each of dimension at most $2n/3$
  such that $B(V,c\log n)\subseteq \bigcup_{i=1}^m V_i$.
\item The algorithm then finds an $x\in\field_2^n$ such that $x\notin
  \bigcup_i V_i$. (This is done using a method similar to the method
  of pessimistic estimators introduced by Raghavan \cite{Rag88}.)
\end{enumerate}

Our algorithm will proceed exactly as the above algorithm in the first
step. The second step of our algorithm will be different (assuming
that the group $\mcG$ is Abelian). We first state Step 1 of the
algorithm of \cite{APY09} in greater generality:

\begin{lemma}
	\label{lemma_cover_subgroups}
	Let $\mcG$ be any fixed finite group with $|\mcG|\geq 2$. For any
	constant $c>0$ and large enough $n$, the following holds. Given any
	subgroup $\mcH$ of $\mcG^n$ such that $\delta(\mcH)\leq
	\frac{n}{2}$, there is a collection of $m \leq n^{10c}$ subgroups
	$\mcH_1,\mcH_2,\ldots,\mcH_m$ such that $B(\mcH,c\log n)\subseteq
	\bigcup_{i=1}^m\mcH_i$, and $\delta(\mcH_i)\leq 2n/3$ for each $i$.
	Moreover, there is a logspace algorithm that, when given as input
	$\mcH$ as a set of generators, produces generators for the subgroups
	$\mcH_1,\mcH_2,\ldots,\mcH_m$.
\end{lemma}

\begin{proof}
	The proof follows exactly as in \cite{APY09}. We reproduce it here
	for completeness and to analyze the complexity of the procedure. 

	Let $1$ denote the identity element of $\mcG$. For each $S\subseteq
	[n]$, let $\mcG^n(S)$ denote the subgroup of $\mcG^n$ consisting of
	those $x$ such that $x_i = 1$ for each $i\notin S$. Note that
	$\delta(\mcG^n(S)) = |S|$. Also note that for each $S\subseteq
	[n]$, the group $\mcG^n(S)$ is a normal subgroup; in particular,
	this implies that the set $\mc{K}\cdot\mcG^n(S)$ is a subgroup of
	$\mcG^n$ whenever $\mc{K}$ is a subgroup of $\mcG^n$.
	
	Partition the set $[n]$ into $\ell\leq 10c\log n$ sets of
	size at most $\lceil\frac{n}{10c\log n} \rceil$ each -- we will call
	these sets $S_1,S_2,\ldots,S_\ell$. For each $A\subseteq [\ell]$ of
	size $\lceil c\log n \rceil$, let $\mc{K}_A$ denote the subgroup
	$\mcG^n(\bigcup_{i\in A} S_i)$. Note that the number of such
	subgroups is at most $2^\ell \leq n^{10c}$. Also, for each $A$ as
	above, $\delta(\mc{K}_A) = |\bigcup_{i\in A}S_i| \leq
	\left(\frac{n}{10c\log n} + 1\right)(c\log n+1) < \frac{n}{9}$, for
	large enough $n$.

	Consider any $x\in B(c\log n)$ (i.e, an element $x$ of $\mcG^n$ s.t $wt(x)\leq
	c\log n$). We know that $x\in \mcG^n(S)$ for some $S$ of size at
	most $c\log n$. Hence, it can be seen that $x\in
	\mcG^n(\bigcup_{i\in A}S_i)$ for some $A$ of size $\lceil c\log
	n\rceil$; this shows that $B(c\log n)\subseteq \bigcup_{A}\mc{K}_A$.
	Therefore, we see that $B(\mcH,c\log n) = \mcH B(c\log n)\subseteq
	\bigcup_{A} \mcH \mc{K}_A$. 

	For each $A\subseteq [\ell]$ of size $\lceil c\log n\rceil$, let
	$\mcH_A$ denote the subgroup $\mcH \mc{K}_A$ (note that this is
	indeed a subgroup, since $\mc{K}_A$ is a normal subgroup).
	Moreover, the cardinality of this subgroup is bounded by
	$|\mcH|\cdot |\mc{K}_A|\leq |\mcG|^{n/2}|\mcG|^{n/9} <
	|\mcG|^{2n/3}$; hence, $\delta(\mcH_A)\leq 2n/3$. Thus, the
	collection of subgroups $\bra{\mcH_A}_A$ satisfies all the
	properties mentioned in the statement of the lemma. That a set of
	generators for this subgroup can be computed in deterministic
	logspace -- for some suitable choice of $S_1,S_2,\ldots,S_\ell$ --
	is a routine check from the definition of the subgroups
	$\bra{\mc{K}_A}_A$. This completes the proof of the lemma.
\end{proof}

Using Lemma \ref{lemma_cover_subgroups}, we are able to efficiently
``cover'' $B(\mcH, c\log n)$ for any small subgroup $\mcH$ of $\mcG^n$
by a union of small subgroups. Therefore, to find a point that
is $c\log n$-far from $\mcH$, it suffices to find a point $x\in\mcG^n$ not
contained in any of the covering subgroups. To do this, we note that
if $S$ is a multiset containing elements from $\mcG^n$ such that
$C(\mcG^n, S)$ is a Cayley graph with good expansion, then $S$ must
contain such an element. This is formally stated below.

\begin{lemma}
	\label{lemma_exp_hitting_set}
	For any constant $c > 0$ and large enough $n\in\naturals$, the
	following holds. Let $S$ be any multiset of elements of $\mcG^n$
	such that $\lambda(C(\mcG^n, S)) < \frac{1}{n^{20c}}$. Then, for
	$m \leq n^{10c}$ and any collection $\mcH_1,\mcH_2,\ldots,\mcH_m$ of
	subgroups such that $\delta(\mcH_i)\leq 2n/3$ for each $i$, there is
	some $s\in S$ such that $s\notin \bigcup_i \mcH_i$.
\end{lemma}

\begin{proof}
	The proof follows easily from Lemma \ref{lemma_exp_set_subgps}.
	Given any $i\in [m]$, we know, from Lemma
	\ref{lemma_exp_set_subgps}, that $|S\intersect \mcH_i| <
	\frac{|S|}{n^{10c}}$ (where the elements of the multisets are
	counted with repetitions). Hence, $|S\intersect
	\bigcup_i\mcH_i|\leq \sum_i |S\intersect \mcH_i| <
	\frac{m|S|}{n^{10c}}\leq |S|$. Therefore, there must be some $s\in
	S$ such that $s\notin \bigcup_i\mcH_i$.
\end{proof}

Therefore, to find a point $x$ that is $c\log n$-far from the subspace
$\mcH$, it suffices to construct an $S$ such that $C(\mcG^n, S)$ is a
sufficiently good expander, find the covering subgroups $\mcH_i$
($i\in[m[$), and then to find an $s\in S$ that does not lie in any of the
$\mcH_i$. We follow the above approach to give an efficient parallel
algorithm for the RPP in the case that $\mcG$ is an Abelian group. For
arbitrary groups, we show that the method of \cite{APY09} yields a
polynomial time algorithm.

\section{Remote Point Problem for Abelian Groups}

Fix an Abelian group $\mcG$. Recall that a \emph{character} $\chi$ of
$\mcG^n$ is a homomorphism from $\mcG^n$ to $\complex^*_1$, the
multiplicative subgroup of the complex numbers of absolute value $1$.
For $\varepsilon>0$, a distribution $\mu$ over $\mcG^n$ is said to be
$\varepsilon$-biased if, given any non-trivial character $\chi$ of
$\mcG^n$,
$
\left|\avg{x\sim \mu}{\chi(x)}\right|\leq \varepsilon
$.

A multiset $S$ consisting of elements from $\mcG^n$ is said to be an
\emph{$\varepsilon$-biased space in $\mcG^n$} if the uniform
distribution over $S$ is an $\varepsilon$-biased distribution.

It can be checked that a multiset consisting of
$(\frac{n}{\varepsilon})^{O(1)}$ independent, uniformly random
elements from $\mcG^n$ form an $\varepsilon$-biased space with high
probability. Explicit $\varepsilon$-biased spaces were constructed for
the group $\field_2^n$ by Naor and Naor in \cite{NN93}; further
constructions were given by Alon et al. in \cite{AGHP92}.  Explicit
constructions of $\varepsilon$-biased spaces in $\integer_d^n$ were
given by Azar et al. in \cite{AMN98}. We observe that this last
construction yields a construction for all Abelian groups $\mcG^n$,
when $\mcG$ is of constant size. We first state the result of
\cite{AMN98} in a form that we will find suitable.

\begin{theorem}\label{thm_AMN_construction}
  For any fixed $d$, there is an $\NC^2$ algorithm that does the
  following. On input $n$ and $\varepsilon>0$ (both in unary), the
  algorithm produces a symmetric multiset $S\subseteq\integer_d^n$ of
  size $O( (\frac{n}{\varepsilon})^2)$ such that $S$ is an
  $\varepsilon$-biased space in $\integer_d^n$.
\end{theorem}

\begin{proof}
  It is easy to see that the $\varepsilon$-biased space construction
  in \cite{AMN98} can be implemented in deterministic logspace (and
  hence in $\NC^2$). If the space $S$ obtained is not symmetric, we can
	consider the multiset that is the disjoint union of $S$ and
	$S^{-1}$, which is also easily seen to be $\varepsilon$-biased.
\end{proof}

\begin{remark}
	We note that the definition of small bias spaces in \cite{AMN98}
	differs somewhat from our own definition above. But it is easy to
	see that an $\varepsilon$-bias space in $\integer_d^n$ in the sense
	of \cite{AMN98} is a $(d\varepsilon)$-bias space according to our
	definition above.
\end{remark}

\begin{remark}
  In a recent paper, Meka and Zuckerman \cite{MZ09} observe, as we do
  below, that the construction of \cite{AMN98} gives small bias spaces
  for any arbitrary Abelian group $\mcG$.  Nevertheless, we present
  our own proof of this fact, since the small bias spaces that follow
  from our proof are of \emph{smaller} size.  Specifically, our proof
  shows how to explicitly construct sample spaces of size $O\left(
    \frac{n^2}{\varepsilon^2} \right)$, whereas the relevant result in
  \cite{MZ09} only produces small bias spaces of size $O\left(
    (\frac{n}{\varepsilon})^b \right)$, where $b$ is some constant
  that depends on $\mcG$ (and can be as large as $\Omega(\log
	|\mcG|)$).
\end{remark}

\begin{lemma}\label{corollary_gp_spaces}
  For any fixed group $\mcG$, there is an $\NC^2$ algorithm which, on
  input $n$ and $\varepsilon>0$ in unary, produces a symmetric
  multiset $S\subseteq\mcG^n$ of size $O((\frac{n}{\varepsilon})^2)$
  such that $S$ is an $\varepsilon$-biased space in $\mcG^n$.
\end{lemma}

\begin{proof}
  By the Fundamental Theorem of finite Abelian groups, $\mcG\cong
  \integer_{d_1}\oplus \integer_{d_2}\oplus
  \cdots\oplus\integer_{d_k}$, for positive integers
  $d_1,d_2,\ldots,d_k$ such that $d_1\mid d_2\mid\cdots\mid d_k$. Let
  $\mcG_0$ denote $\integer_{d_k}^k$. Note that for any
  $s,t\in\naturals$, $\integer_s\cong \integer_{st}/\integer_t$.
  Hence, we see that that $\mcG \cong \mcG_0/\mcH$, where $\mcH$ is
  the subgroup
  $\integer_{e_1}\oplus\integer_{e_2}\oplus\cdots\oplus\integer_{e_k}$,
  and $e_i = d_k/d_i$ for each $i\in [k]$. Therefore, $\mcG^n\cong
  \mcG_0^n/\mcH^n$. Let $\pi:\mcG_0^n\rightarrow \mcG^n$ be the
  natural onto homomorphism with kernel $\mcH^n$. Note that $\pi$ is
  just the projection map and can easily be computed in $\NC^2$.

  Since $\mcG_0^n\cong \integer_{d_k}^{nk}$, by Theorem
  \ref{thm_AMN_construction}, there is an $\NC^2$ algorithm that
  constructs a symmetric multiset $S_0\subseteq\mcG_0^n$ of size
  $O(\left(\frac{kn}{\varepsilon}\right)^2)$ such that $S_0$ is an
  $\varepsilon$-biased space in $\mcG_0^n$. We claim that the multiset
  $S = \pi(S_0)$ is a symmetric $\varepsilon$-biased space in 
  $\mcG^n$. To see this, consider any non-trivial character $\chi$ of
  $\mcG^n$; note that $\chi_0 = \chi\circ\pi$ is a non-trivial
  character of $\mcG_0^n$. We have
  \[
  \left|\avg{x\sim S}{\chi(x)}\right| = \left|\avg{x_0\sim
      S_0}{\chi(\pi(x_0))}\right| = \left|\avg{x_0\sim
      S_0}{\chi_0(x)}\right| \leq \varepsilon
  \]
	where the first equality follows from the definition of $S$, and the
	last inequality follows from the fact that $S_0$ is an
	$\varepsilon$-biased space in $\mcG_0^n$. Since $\chi$ was an
	arbitrary non-trivial character of $\mcG^n$, we have proved that $S$
  is indeed an $\varepsilon$-biased space in $\mcG^n$. It is easy to
  see that $S$ is symmetric. Finally, note that $S$ can be computed in
  $\NC^2$. This completes the proof.
\end{proof}

Finally, we mention a well-known connection between small bias spaces
in $\mcG^n$ and Cayley graphs over $\mcG^n$ (e.g.\ see Alon and
Roichman \cite{AR94}).

\begin{lemma}\label{lemma_e_bias_exp}
  Given any symmetric multiset $S\subseteq\mcG^n$, the Cayley graph
  $C(\mcG^n, S)$ is an $(|\mcG|^n, |S|, \alpha)$-graph iff $S$ is an
  $\alpha$-biased space.
\end{lemma}

Lemmas \ref{lemma_e_bias_exp} and \ref{corollary_gp_spaces} have the
following easy consequence:
\begin{lemma}
  \label{corollary_exp_groups}
  For any Abelian group $\mcG$, there is an $\NC^2$ algorithm which, on unary
  inputs $n$ and $\alpha > 0$, produces a symmetric multiset
  $S\subseteq\mcG^n$ of size $O( (\frac{n}{\alpha})^2)$ such that
  $C(\mcG^n,S)$ is a $(|\mcG|^n,|S|,\alpha)$-graph.
\end{lemma}
Putting the above statement together with the results of Section
\ref{section_main_lemma}, we have the following.

\begin{theorem}\label{thm_rpp_abelian}
  For any constant $c>0$, the RPP over $\mcG$ has an $\NC^2$
  $(n/2,c\log n)$-algorithm.
\end{theorem}

\begin{proof}
  Let $\mcH$ denote the input subgroup. By Lemma
  \ref{lemma_cover_subgroups}, there is a logspace (and hence $\NC^2$)
  algorithm that computes a collection of $m = n^{O(c)}$ many
  subgroups $\mcH_1,\mcH_2,\ldots,\mcH_m$ such that $B(\mcH,c\log
  n)\subseteq\bigcup_{i=1}^m\mcH_i$ and $\delta(\mcH_i)\leq 2n/3$ for
  each $i\in [m]$. Now, fix any multiset $S\subseteq\mcG^n$ such that
  the Cayley graph $C(\mcG^n, S)$ is a $(|\mcG|^n, |S|,\alpha)$-graph,
  where $\alpha = \frac{1}{2n^{20c}}$; by Lemma
  \ref{corollary_exp_groups}, such an $S$ can be constructed in
	$\NC^2$. It follows from Lemma \ref{lemma_exp_hitting_set} that
	there is some $s\in S$ such that $s\notin\bigcup_{i=1}^m\mcH_i$.
	Finally, by Lemma \ref{corollary_gp_membership}, there is an $\NC^2$
	algorithm to test if each $s\in S$ belongs to $\mcH_i$, for any
	$i\in[m]$. Hence, we can find out (in parallel) exactly which $s\in
	S$ do not belong to any of the $\mcH_i$ and output one of them. The
	output element $s$ is surely $c\log n$-far from $\mcH$.
\end{proof}


Let $\mcG$ be Abelian. We observe that a method of \cite{APY09},
coupled with Theorem~\ref{thm_rpp_abelian}, yields an efficient
$(k,\frac{cn\log k}{k})$-algorithm for any constant $c > 0$, and
$k\leq n/2$.

\begin{theorem}\label{thm_arb_k_abelian}
  Let $c > 0$ be any constant. If $\mcG$ is an Abelian group, then the
  RPP over $\mcG$ has an $\NC^2$ $(k,\frac{cn\log k}{k})$-algorithm
  for any $k\leq n/2$.
\end{theorem}

\begin{proof}
  Given as input a subgroup $\mcH$ such that $\delta(\mcH)=k\leq n/2$,
  the algorithm partitions $[n]$ as $[n] = \bigcup_{i=1}^m T_i$, where
  $2k\leq |T_i| < 4k$ for each $i$; note that $m\geq
  n/4k$. Let $\mcH_i$ denote the subgroup obtained when $\mcH$ is
  projected onto the coordinates in $T_i$. Since $\delta(\mcH_i)\leq
  k\leq |T_i|/2$, we can, by Theorem \ref{thm_rpp_abelian},
  efficiently find a point $x_i\in\mcG^{|T_i|}$ that is at least
  $4c\log k$-far from $\mcH_i$.  Putting these $x_i$ together in the
  natural way, we obtain an $x\in\mcG^n$ that is $\frac{cn\log
    k}{k}$-far from the subgroup $\mcH$.

  Since $\mcG$ is Abelian, using the algorithm of
  Theorem~\ref{thm_rpp_abelian}, the $x_i$ can all be computed in
  parallel in $\NC^2$. Hence, the entire procedure can be performed in
  $\NC^2$.
\end{proof}

\section{RPP over General Groups}

Let $\mcG$ denote some fixed finite group. We can generalize the
polynomial-time algorithm of \cite{APY09}, described for $\F_2$, to
compute a point $x\in\mcG^n$ that is $c\log n$-far from a given input
subgroup $\mcH$ such that $\delta(\mcH)\leq n/2$. We only state this
result below and refer the interested reader to the full version
\cite{AS09b} for details.

\begin{theorem}\label{thm_rpp_non_abelian}
  For any constant $c > 0$, the RPP over $\mcG$ has a polynomial time
  $(n/2, c\log n)$-algorithm.
\end{theorem}

Analogous to Theorem~\ref{thm_arb_k_abelian}, we have the following
solution to RPP for general groups.

\begin{theorem}\label{thm_arb_k}
  Let $c>0$ be any constant.  For any $\mcG$, the RPP over $\mcG$ has
  a polynomial time $(k,\frac{cn\log k}{k})$-algorithm for any $k\leq
  n/2$.
\end{theorem}

\begin{proof}
  The construction is exactly the same as in the proof of
  Theorem~\ref{thm_arb_k_abelian}. The only difference is that we will
  apply the algorithm of Theorem \ref{thm_rpp_non_abelian}. In this
  case, the $x_i$ can all be found in deterministic polynomial time.
  Hence, the entire procedure gives us a polynomial-time algorithm.
\end{proof}

\section{Limitations of expanding sets}

In the previous sections, we have shown how generators for expanding
Cayley graphs on $\mcG^n$, where $\mcG$ is a fixed finite group, can
help solve the RPP over $\mcG$. In particular, we have the following
easy consequence of Lemmas \ref{lemma_cover_subgroups} and
\ref{lemma_exp_hitting_set}.

\begin{corollary}\label{corollary_exp_hitting_set}
  For any constant $c > 0$, large enough $n$, and any symmetric multiset
  $S\subseteq\mcG^n$ such that $\lambda(C(\mcG^n,S)) <
  \frac{1}{n^{20c}}$, the following holds. If $\mcH$ is any subgroup
  of $\mcG^n$ such that $\delta(\mcH)\leq n/2$, there is some $s\in S$
  such that $s\notin B(\mcH,c\log n)$.
\end{corollary}

It makes sense to ask if the parameters in Corollary
\ref{corollary_exp_hitting_set} are far from optimal. Is it true that
any polynomial-sized symmetric multiset $S\subseteq\mcG^n$ with good
enough expansion properties is $\omega(\log n)$-far from every
subgroup of dimension at most $n/2$? We can show that this is not true.
Formally, we can prove:

\begin{theorem}\label{thm_exp_limits}
  For any constant $c>0$ and large enough $n$, there is a symmetric
  multiset $S\subseteq\field_2^n$ such that $\lambda(C(\field_2^n,
  S))\leq \frac{1}{n^c}$ but there is a subspace $L$ of dimension
  $n/2$ such that $S\subseteq B(L,20c\log n)$.
\end{theorem}

It is well known that for any family of $d$-regular multigraphs $G$
$\lambda(G)=\Omega(1/\sqrt{d})$ (see e.g.\ \cite[Theorem 5.3]{HLW06}).
As a consequence of this lower bound it follows for any fixed group
$\mcG$ and any multiset $S\subseteq\mcG^n$ that $\lambda(C(\mcG,S))=
\Omega(1/\sqrt{|S|})$. Hence, the above theorem tells us that just
the expansion properties of $C(\field_2^n,S)$ for any $\poly{n}$-sized
$S$ are not sufficient to guarantee $\omega(\log n)$-distance from
every subspace of dimension $n/2$. The proof of the above statement
can be found in the full version \cite{AS09b}.
\section{Discussion}
For the remote point problem over an Abelian group $\mcG$, we have
shown how expanding generating sets for Cayley graphs of $\mcG^n$ can
be used to obtain deterministic $\NC^2$ algorithms. A natural question
is whether we can obtain a similar algorithm for non-Abelian $\mcG$.
Note that Lemma \ref{lemma_exp_hitting_set} holds in the non-Abelian
setting too. Hence, in order to obtain an $\NC^2$-algorithm for the
RPP over arbitrary non-Abelian $\mcG$ along the lines of our algorithm
for Abelian groups, we need to be able to check (in $\NC^2$) for
membership in $\mcG^n$, and we need to be able to construct small
multisets $S$ of $\mcG^n$ such that $C(\mcG^n,S)$ has sufficiently
good expansion properties. Luks' work \cite{Lu86} yields an $\NC^4$
test for membership in $\mcG^n$ for arbitrary $\mcG$.  Building on
that, there is also an $\NC^2$ membership test for $\mcG^n$
\cite{AKV05}. However, we are unable to compute a (good enough)
expanding generator set for the group $\mcG^n$ in deterministic $\NC$
or even in deterministic polynomial time.

\section*{Acknowledgements}
We are grateful to Noga Alon and Sergey Yekhanin for interesting
comments. In particular, Alon pointed out to us that
Lemma~\ref{lemma_exp_set_subgps} has an alternative proof using the
expander mixing lemma. We thank the anonymous referees for their
comments and suggestions.

\end{document}